%% file: deterministic-subsets.tex
\documentclass{article}
\title{Learning What's Going On: Reconstructing Preferences and Priorities from Opaque Transactions}

\let\chapter\section
\usepackage[affil-it]{authblk}
\usepackage{fullpage}
\usepackage{amsthm, amsmath, amsfonts, amssymb, xspace}
\usepackage[ruled,vlined]{algorithm2e}
\usepackage{color}
\usepackage{xspace}
\usepackage{tikz}
\usetikzlibrary{arrows,decorations,decorations.shapes,backgrounds,shapes}

\newcommand{\argmax}{\textrm{argmax}}

\newtheorem{thm}{Thm}[section]
\newtheorem{lemma}[thm]{Lemma}
\newtheorem{corollary}[thm]{Corollary}

\newtheorem{observation}[thm]{Observation}

\usepackage{color-edits}
\usepackage{hyperref}
\usepackage{natbib}

\newcommand{\s}{\succ}

\newcommand{\sh}{\widehat{\succ}}
\newcommand{\sg}{\tilde{\succ}}

\addauthor{ab}{red}
\addauthor{ym}{cyan}
\addauthor{jm}{magenta}
\newcommand{\M}{\ensuremath{\mathcal{A}}\xspace}

\newcommand{\MB}{\ensuremath{\mathcal{MB}}\xspace}
\newcommand{\demt}{\ensuremath{\textit{d}^t}\xspace}
\newcommand{\perm}[1]{\ensuremath{\texttt{Permutation}(#1)\xspace}}
\newcommand{\init}{\ensuremath{\texttt{InitPerm}\xspace}}
\newcommand{\demote}[2]{\ensuremath{\texttt{Demote}(#1, #2)\xspace}}
\newcommand{\ith}[3]{\ensuremath{\texttt{Buyer}(#1, #2, #3)}\xspace}
\newcommand{\loc}[2]{\ensuremath{\texttt{Loc}(#1, #2)\xspace}}
\newcommand{\level}[2]{\ensuremath{\texttt{Level}(#1, #2)\xspace}}

\newcommand{\firstmistake}[3]{\ensuremath{\texttt{FirstMistake}(#1,#2,#3)}\xspace}

\author[1]{Avrim Blum\thanks{Supported in part by the National Science Foundation under grants CCF-1101215, CCF-1116892, CCF-1331175, and IIS-1065251.  Email: {\tt avrim@cs.cmu.edu}}}

\author[2]{Yishay Mansour\thanks{This research was supported in part by The Israeli Centers of Research Excellence (I-CORE) program,
(Center  No. 4/11), by a grant from the Israel Science Foundation (ISF), by a grant from United
  States-Israel Binational Science Foundation (BSF), and by a grant
  from the Israeli Ministry of Science (MoS). Email: {\tt mansour@tau.ac.il}}}
\author[1]{Jamie Morgenstern\thanks{Supported in part by the National Science Foundation under grants  CCF-1116892, CCF-1331175 and IIS-1065251 and by a Simons Award for Graduate Students in Theoretical Computer Science.  Email: {\tt jamiemmt@cs.cmu.edu}}}
\affil[1]{Computer Science Department, Carnegie Mellon University}
\affil[2]{Tel Aviv University and Microsoft Research, Hertzelia}

\begin{document}
\maketitle

\begin{abstract}
  We consider a setting where $n$ buyers, with combinatorial
  preferences over $m$ items, and a seller, running a priority-based
  allocation mechanism, repeatedly interact.  Our goal, from observing
  limited information about the results of these interactions, is to
  reconstruct both the preferences of the buyers and the mechanism of
  the seller.  More specifically, we consider an online setting where
  at each stage, a subset of the buyers arrive and are allocated
  items, according to some unknown priority that the seller has among
  the buyers. Our learning algorithm observes only which buyers arrive
  and the allocation produced (or some function of the allocation,
  such as just which buyers received positive utility and which did
  not), and its goal is to predict the outcome for future subsets of
  buyers. For this task, the learning algorithm needs to reconstruct
  both the priority among the buyers and the preferences of each
  buyer.  We derive mistake bound algorithms for additive, unit-demand
  and single minded buyers. We also consider the case where buyers'
  utilities for a fixed bundle can change between stages due to
  different (observed) prices.  Our algorithms are efficient both in
  computation time and in the maximum number of mistakes (both
  polynomial in the number of buyers and items).
\end{abstract}

\thispagestyle{empty}
\newpage
\clearpage
\setcounter{page}{1}

\input{intro}

\input{results}
\input{rw}

\input{prelims}
\input{single}
\input{unit}

\input{prices}

\input{discussion}

\bibliographystyle{plainnat}
\bibliography{deterministic-subsets.bib}

\end{document}

%% file: intro.tex
\section{Introduction}
A collection of lobbyists enter a politician's office.  An hour later,
they emerge, some happy and some unhappy.  The next day, a different
subset of lobbyists enter, and again some emerge happy and some unhappy.
Suppose that what is happening is that the politician has a collection
of $m$ favors (items) to distribute, along with a priority ordering
over lobbyists; the lobbyists are single-minded, each lobbyist $i$
with a demand-set $D_i \subseteq \{1,\ldots,m\}$.  The politician
orders the lobbyists who arrived that day by priority and hands each
one $D_i$ if it is still available (making the lobbyist happy) or
handing her nothing if $D_i$ is no longer available (making the
lobbyist unhappy).  Can we reconstruct the politician's priority
ordering and the lobbyists' demand-sets (or at least, given a set of
lobbyists, predict which will end up happy and which unhappy, since we
cannot observe the items themselves) from these types of observations?

Or, in the context of computational advertising,
consider a publisher that owns a web site and has some collection of
advertisers. The advertisers tell the publisher which potential impressions are
relevant to them (are they interested in this type of impression), their bid (the value they are 
willing to pay for a relevant impression) 
and conflicts (which competing advertisers they refuse to appear concurrently with, 
for example, competing car makers for a car ad).
Each time a user visits a webpage, the publisher's ad server considers the subset of 
relevant advertisers. It then orders the advertisers (say, by their bid) and greedily 
assigns an impression to an advertiser if it does not introduce a conflict (otherwise 
it skips this advertiser).
%
From observing which
advertisements are shown and which are not, and knowing which advertisers are relevant,
can we learn the conflicts
and priority ordering?

In this paper, we consider this and several closely related problems.
Formally, we assume there are $n$ buyers (lobbyists or advertisers)
and a mechanism (the politician or ad server) who has a
priority ordering over buyers that is unknown to us.  There is a
collection of $m$ items, and the buyers each have utility functions over
subsets of items (e.g., the examples above correspond to the case of
single-minded buyers\footnote{In the case of advertisers, each
  pairwise conflict can be modeled as an abstract item that belongs to
  the demand-set of both conflicting advertisers.}).  At each
time-step $t$, some set $S^t \subseteq \{1,\ldots, n\}$ of buyers
arrive.  The mechanism then orders the buyers in $S^t$ by priority and
allocates to each its most-preferred bundle from the collection of
items not yet given to earlier buyers in the ordering.  Finally, we
observe some function $y^t$ of the outcome (allocation).  We will consider the case
that buyers are single-minded and $y^t$ indicates which buyers received
positive utility and which did not (as in the examples above), as well
as the case that buyers are unit-demand or additive, and $y^t$
indicates the items (if any) that each buyer received.  The goal of
our algorithm will be to predict $y^t$ from $S^t$, and we will present efficient
algorithms that can do so while making only a bounded (polynomial in
$n$ and $m$) number of mistakes in total.

Notice that the setting of single-minded buyers can exhibit
significant non-monotonicities.  For example, consider two lobbyists
$a$ and $b$ whose demand sets $D_a$ and $D_b$ do not overlap, and with
$a$ having higher priority than $b$.  Depending on the other lobbyists
who show up, it could be that $a$'s presence has no effect on $b$
(since their sets don't overlap); it could be that $a$'s presence {\em
  helps} $b$ (if there is a lobbyist $c$ present, with priority
between $a$ and $b$, such that $D_a \cap D_c \neq \emptyset$ and $D_c
\cap D_b \neq \emptyset$); or, it could be that $a$'s presence {\em
  hurts} $b$ (if there are lobbyists $d,e$ present with $D_a \cap D_d
\neq \emptyset$, $D_d \cap D_e \neq \emptyset$, and $D_e \cap D_b \neq
\emptyset$, with ordering $a\s d \s e \s b$).

To get a feel for the type of results we are aiming for, we describe
here a simpler case of this problem and how one can solve it. Suppose
buyers are additive rather than single-minded,\footnote{For each buyer
  $i$ and each item $j$, either buyer $i$ either wants $j$ or she
  doesn't, and buyer $i$ takes all items that she wants that are
  available when it is her turn.  One can think of this as the
  behavior of additive buyers in the presence of fixed prices.}  and
$y^t$ denotes the \emph{allocation} of items to the agents in
$S^t$. This problem is monotonic in some sense: if $S^{t'} \subseteq
S^t$ and $i\in S^{t'}$, then $y_i^t \subseteq y_i^{t'}$ (including
more buyers reduces the allocation for $i$). It is possible to solve
this problem tracking two things: first, for a given buyer $i$, track
the set of items $i$ has ever won, and second, an estimate the
relative ordering of the buyers $\sh$. Consider some item $j$ that buyer $i$
wins in some round. Buyer $i$ will take item $j$ whenever it is still available,
so if buyer $i$ doesn't win item $j$, we learn that buyer $i$ is later in the ordering
than the winner of item $j$.  To predict the allocation for a set $S^t$, we
order buyers in $S^t$ according to $\sh$, and in that order, give the
buyers all of the remaining items she has bought before.

Our algorithm will make two types of mistakes, and we can limit the
number of each.  Consider the first buyer (according to $\sh$) for
whom we make a mistake in predicting her allocation.  Suppose she
won some item $j$ that we did not predict she would get. Since she was
the first mistake according to $\sh$, we did not predict that someone
earlier in $\sh$ won item $j$. Thus, item $j$ was available in our prediction
when we reached buyer $i$: we did not allocate item $j$ to her because we had never
seen her win item $j$ before. There are at most $n m$ of these sorts of
mistakes to make (one per item/buyer pair). Suppose instead we predicted
that some item $j$ would be allocated to buyer $i$ but she did not get item $j$. Since
we predicted item $j$ for buyer  $i$, buyer $i$ must have won item $j$ before (and is
therefore interested in item $j$). Then, it must be the case that buyer $i$ is
later in the true ordering than in $\sh$: the winner of item $j$ must be
earlier than buyer $i$. Then, we can update $\sh$ by demoting buyer $i$: if done
carefully, as we describe in Section~\ref{sec:single}, $i$ will never
be demoted further than her true position in the ordering, so there
will be at most $n^2$ mistakes of this type.

%% file: results.tex
\subsection{Our Results}

This paper presents several mistake-bound learning algorithms for
ordered arrival mechanisms. The crux of these algorithms is to learn
the hidden \emph{priority order}, or permutation over buyers, in a way
that meshes well with learning the players' preferences at the same
time, all in a mistake-bound framework. First, we consider the case
without prices (or equivalently, when prices are fixed across
time). In the case of a single item, this problem reduces to learning
the priority order over buyers. Previous work describes how one can
efficiently sample linear extensions of partial
orders~\citep{karzanov1991conductance}, which can be combined with a
simple halving algorithm to learn a permutation with a mistake bound
of $\Theta(n\log(n))$ when mistakes are accompanied with some pair
$i,j$ which were mis-ordered (see Section \ref{sec:discussion}). When buyers have more general
valuations, however, it is not clear how to use this algorithm to learn the priority order
over buyers.\footnote{In the case of a single item, we learn that the true
  winner has higher priority than everyone else. In general, mistakes
  don't give such a simple constraint on the ordering of
  buyers.} So, we use a simpler algorithm for learning a permutation
whose mistake bound is $\Theta(n^2)$ (when a mistake is accompanied by
some element of the permutation that needs to be demoted rather than a
pair of elements for which the permutation was incorrect).  With this
algorithm in hand, we build mistake-bound learning algorithms for
single-minded, unit-demand, and additive buyers with fixed prices,
and for unit-demand and additive buyers with variable, observable
prices. The precise form of these bounds is summarized in
Table~\ref{table:results}. The results for additive and unit-demand
buyers also apply to the case where there are multiple copies of
goods.

\begin{figure}
\begin{center}
\begin{tabular}{|c|c|c|}
\hline
Type of buyers & Prices (fixed or variable) & Mistake Bound\\\hline
Single-minded & fixed & $O(n^2)$\\ \hline
Additive & fixed & $O(nm + n^2)$\\ \hline
Unit-demand & fixed & $O(n^2m\log(m))$\\ \hline
Additive & variable & $O(n^2m\log(V))$\\ \hline
Unit-demand & variable & $O(n^2\MB)$\\ \hline
\end{tabular}
\end{center}
\caption{Summary of our results; $V$ is the maximum value any buyer
  has for an item and $\MB$ is the mistake bound for the Ellipsoid
  algorithm. Our algorithm for single-minded buyers applies even to
  the case where observations are only which buyers get their set and
  which do not, rather than the explicit
  allocation.\label{table:results} }
\end{figure}

These results look quite a bit like combining online decision list
learning \citep{warmuth1990lists} combined with other forms of
mistake-bound learning algorithms: in the case of unit-demand without
prices, we learn the decision list order over buyers and, for each
buyer, their preference order over items; in the case of additive
buyers without prices, we learn the order over buyers and which
items each person is interested in; in the settings with prices, we
use binary search or ellipsoid-style learning algorithms to understand
the players' numeric values.

%% file: rw.tex
\subsection{Related Work}

Our work is related to the literature on learning from revealed
preferences~\citep{samuelson}\citep{varian2006revealed}, which
considers the problem of learning about a single buyer from observing
her behavior under observed prices.  Rationalizable demands (according
to some prices) are those which can arise from maximizing some
concave, monotone, continuous value function subject to a budget
constraint. \citet{beigman2006learning} gave algorithms with finite
sample complexity bounds for learning and predicting from
rationalizable demand/price pairs.~\citet{zad2012} gave
computationally efficient versions of these algorithms for linear and
linearly separable utilities.~\citet{amin2014revealed} consider both
the problem of setting prices (minimizing regret w.r.t revenue) and
the prediction problem (minimizing classification mistakes
w.r.t. exogenous prices) for the online version of this
problem.~\citet{balcan2014revealed} give tight sample complexity
results for predicting linear, separable piecewise linear concave,
CES, and Leontif preferences, and extend the results to the agnostic
setting. Their results extend to indivisible goods and certain
nonlinear pricing.

Our work differs from the work on revealed preferences in two key
ways. One is that in the settings we consider, each observed
transaction involves multiple buyers: as a result, a buyer might not
get some item because it was not in her desired bundle, or might not
get it because the item (or a complement to it) was taken by some other
buyer. The other is that we also, at the same time, are aiming to
learn an unknown priority ordering for the seller.  The algorithms we
derive use as subroutines algorithms that are closely related to those
for learning decision lists in a mistake bound
setting~\citet{warmuth1990lists}. We also employ a mistake bound
learner for classification by halfspaces from~\citet{maass1990fast}
for efficiently learning a classifier for a fixed unit-demand buyer.

%% file: prelims.tex
\section{Model and Preliminaries}

Let $B$ be a set of $n$ buyers and $I$ be a set of $m$ items. Each
buyer $i\in B$ has some combinatorial valuation $v_i$ over subsets of
items $T\subseteq I$.An \emph{ordered arrival mechanism} $\M$ consists
of an ordering over buyers $\s$ and allocates items in $I$ as follows.
At each time $t$, an arbitrary subset $S^t$ of the $n$ buyers arrives
online. Then, in order according to $\s$, each buyer in $S^t$ chooses
the bundle from the remaining items that maximizes her value.  So, the
buyer $i_1$ in $S^t$ who is first according to $\s$ chooses the bundle
$X^t_{i_1} \subseteq I$ of maximum value to her, then the buyer $i_2$
in $S^t$ who is second according to $\s$ chooses the bundle $X^t_{i_2}
\subseteq I \setminus X^t_{i_1}$ of maximum value to her, and so on.

\newcommand{\obs}{\texttt{obs}}

The \emph{label} $y^t$ for the example $S^t$ is some function $\obs$
of the allocation $(X^t_1, \ldots, X^t_n) = \M(S^t)$ which arises from
this process. Our goal will be to predict $y^t = \obs(\M(S^t))$ for a
new subset $S^t$ given our previous observations. 
 We will focus on $\obs = Id$ (the identity
function; our goal is to predict the allocation) and $\obs =
(\mathbb{I}[X^t_1 \neq \emptyset], \ldots, \mathbb{I}[X^t_n \neq
\emptyset])$ (the function which indicates which buyers bought at
least one item; our goal is to predict which players have positive
utility).

We will be working in the {\bf mistake-bound model}, where our
learning algorithm $\M$ will progress as follows. In each round, the
algorithm is presented with a subset $S^t$. The algorithm's current
prediction, $\M(S^t)$, will be output.  Then, the algorithm observes
the true label $y^t = \obs(X^t_1, \ldots, X^t_n)$.  If $\M(S^t) \neq
y^t$ (the label predicted is incorrect), round $t$ is counted as a
{\bf mistake}.  The goal in mistake-bound learning is to bound the
worst-case total number of mistakes made over an arbitrarily long
sequence of examples presented to the algorithm. We will call this
worst-case bound the {\bf mistake bound for learning algorithm \M}.

We also consider an extension of this setting, where each example is a
pair $(S^t, p^t)$ of a subset of buyers and a price vector $p^t \in
\mathbb{R}^m$. Then, in order according to $\s$, the buyers in $S^t$
each chooses a bundle from the remaining items maximizing her {\em
  utility}.  We assume that their utility is \emph{quasi-linear in
  money}, i.e., that $u_i(X^t_i, p^t) = v_i(X^t_i) - \sum_{j\in
  X^t_i}p^t(j)$ (where $X^t_i\subseteq I$ represents the bundle player
$i$ chose). We call this model the \emph{variable-price model}; the
previous model, described without prices, can be thought with a fixed
price vector $p$ which does not vary across examples. We also note
that, for several of the problems we study, our algorithms extend to
the case where for each item $e \in I$, there are $k_e$ \emph{copies}
of that item (here, we assume buyers are unit-demand in each item, so
they will never purchase more than one copy of a given item).

%% file: single.tex
\section{Single-Minded Buyers}\label{sec:single}

Suppose there are $n$ single-minded buyers. That is, each buyer $i$
has a single demand set $D_i$ for which she has value $v_i(D_i)>0$.
At each time $t$, a subset $S^t\subseteq[n]$ arrives, and each
buyer in $S^t$ is offered the items which haven't been taken by
earlier buyers. If buyer $i$ is offered some set which includes $D_i$,
she will take $D_i$, otherwise she will take nothing (one can imagine
infinitesimal prices giving a disincentive to 
take excess items). Let $W^t$ denote
the set of buyers who won their demand sets in example $S^t$ (so $W^t
= \obs(\M(S^t))$). From $W^t$ and $S^t$ alone, we wish to be able to
predict the winning set $W^{t'}$ for a new subset $S^{t'}$, where our
performance objective is the number of mistakes we make.

Suppose for a moment our learning algorithm knew $\s$ but not the
demand sets $D_i$. From $\s$ alone, we cannot predict the sets
$W^{t'}$ given $S^{t'}$: we need to understand which buyers have items
in common between their demand sets (e.g., if $D_i \cap D_j \neq
\emptyset$, then $i$ and $j$ will not simultaneously be in any winning
set). For this reason, we will call $i$ and $j$ \emph{in conflict} if
$D_i \cap D_j \neq \emptyset$. Consider the graph $G(V,E)$ where
$V=[n]$ and there is an edge $(i,j)\in E$ iff buyers $i$ and $j$ have
a conflict: we will call this the \emph{conflict graph}.  For each
pair of buyers $(i,j)$, if $i, j\in W^t$ for some $S^t$, the two
buyers clearly do not have a conflict. Notice that, even if $i \s j$,
there will be cases when $j\in W^t$ but $i \notin W^t$ (if $i$ has a
conflict, say, with an earlier winning buyer, but $j$ does not). Thus,
to predict whether buyer $i$ will win, it is not sufficient to know
simply who has precedence over $i$ and who conflicts with $i$: we need
more complete information about the structure of the conflict graph.

Notice that, given the conflict graph $G$ and the ordering over buyers
$\s$, we can predict the winners $W^t$ for an arbitrary subset $S^t$.
In particular, buyer $i$ wins exactly when no one before $i$ won with
whom $i$ has a conflict.  So, we can predict the entire winning set by
scanning through $S^t$ in the order $\s$ and adding buyer $i$ to $W^t$
if no conflicting buyer is already in $W^t$.  Thus, we are done if we
can learn both $G$ and $\s$.

We start by presenting a mistake-bound procedure \ref{alg:all}
(informally alluded to in the introduction) to learn a permutation
$\s$ in a model where each time it makes a mistake on its current
guess $\sh^t$, it is told some item $d^t$ that is incorrectly above
some other item $r$ in $\sh^t$ (but is not told $r$).  The procedure
employs a datastructure $P=(\sg,O)$, where $\sg$ is a permutation and
$O$ is a partition of the buyers to levels.  Given a permutation
$\sigma$, let $\loc{i}{\sigma}$ denote the position of $i$ in $\sigma$
and let $\ith{k}{S}{\sigma}$ denote the name of the $k$th element in
$\sigma$ restricted to $S$.  The datastructure at time $t$ is $P^t$ and
we will associate $P^t$ with its output permutation $\sh^t$, so that
$\loc{i}{P^t} = \loc{i}{\sh^t}$.  The key interface with $P^t$ will be
the function $\demote{d^t}{P^t}$, employed when $d^t$ is returned as a
mistake (an item that is incorrectly above some other item $r$ in
$\sh^t$).  The algorithm is quite similar to that used for learning
decision lists in a mistake bound
setting~\citep{warmuth1990lists}. The following lemma is used
throughout the rest of our analysis.

\begin{algorithm}[t]
\SetAlgoRefName{PermULearn}
\Indp Let $\sg$ be an arbitrary ordering for tiebreaking\;
\Indm * \init, put all items at level $1$ *\\
\Indp $O_1 = \{1, \ldots, n\}$\;
$O_{2, \ldots, n} = \emptyset$\;
$\sh^0 = \sg$\;
\Indm * $\perm{P^t}$, outputs consistent permutation *\\
\Indp
\For{$l = 1$ to $|O|$} {
  $\pi_l$ = Order $O_l$, level $l$ of $O$, according to $\sg $\;
}
 $\sh^t = \pi_1 \cdot \pi_2 \cdot \ldots \pi_l$ the concatenation of the levels' permutations\;\
\Indm * $\demote{i}{P^t}$, demote buyer $i$ *\\
\Indp  Let $k$ be the level in which $i$ resides, e.g. $i \in O_k$\;
Let $O'_k = O_k \setminus \{i\}, O'_{k+1} = O_{k+1} \cup \{i\}$\;
Let $O' = (O_1, \ldots, O_{k-1}, O'_k, O'_{k+1}, O_{k+1}, \ldots)$\;
Let $P^{t+1}  = (\sg, O')$\;
$\sh^{t+1} = \perm{P^{t+1}}$\;
\caption{Maintains level ordering\label{alg:all}}
\end{algorithm}

\begin{lemma}\label{lem:perm-use}
  \ref{alg:all}, below, has a mistake-bound of $O(n^2)$ for learning a
  permutation $\s$ in a model where if $\sh^t \neq \s$, it is given
  some $\demt \in [n]$ such that $\exists r$ such that
$\demt \; \sh^t \; r$ but $r \s d$.
\end{lemma}

\begin{proof}
  Let \level{i}{O} be the level of buyer $i$ in $O$, i.e.,
  $\level{i}{O}=j$ where $i\in O_j$. Let \demt be the element given
  (and demoted by~\ref{alg:all}) at time $t$.
We show by induction on the algorithm's pushing down elements that
no $\demt$ is pushed to a level below her location in $\s$, i.e.,
$\loc{i}{\s}\geq \level{i}{O^t}$. If this is the case, the algorithm
is correct: at most $n^2$ pushes can occur and the limit of these
push-downs is some consistent permutation. Prior to any elements
being pushed down, all elements are at the first level, hence
$\loc{i}{\s}\geq \level{i}{O^0}$ initially.

Now, suppose that the inductive hypothesis holds at time $t-1$,
i.e.,
for any item $i$ we have $\loc{i}{\s} \geq \level{i}{O^{t-1}}$.
%
Our induction hypothesis implies two things: first, that
$\loc{\demt}{\s}\geq \level{\demt}{O^{t-1}}$ and second,
$\loc{r}{\s} \geq \level{r}{O^{t-1}}$

Our assumption states that when $\demt$ is demoted, there is some
element $r$ such that $\loc{\demt}{\s} > \loc{r}{\s}$ but
$\loc{r}{P^t} > \loc{\demt}{P^t} \geq \level{\demt}{O^{t-1}}$. The
second implication of our induction hypothesis states, prior to the
$t$th demotion,
$\loc{r}{\s} \geq \level{r}{O^{t-1}}$. Since $\loc{r}{P^t} >
\loc{\demt}{P^t}$, it must be the case that $\level{r}{O^{t-1}} \geq
\level{\demt}{O^{t-1}}$ ($r$ is only given a later location in $P^t$
than $\demt$ if she is at a weakly larger-numbered level). Thus,
$\loc{\demt}{\s} > \loc{r}{\s} \geq \level{r}{O^{t-1}} \geq
\level{\demt}{O^{t-1}}$, so $\loc{\demt}{\s} > \level{\demt}{O^{t-1}}$
and pushing $\demt$ to level $\level{\demt}{O^{t-1}}+1$ maintains the
invariant. So the algorithm is correct, and no element is pushed more
than $n$ times, implying a mistake bound of $n^2$.
\end{proof}

Thus,~\ref{alg:all} is guaranteed to learn a consistent ordering,
so long as we never request a demotion of a buyer who shouldn't be
demoted. So, when we use this procedure, it suffices to
show that we never demote a buyer $i$ unless there is some buyer $i'$
later in $P^t$ but earlier in $\s$ to guarantee that we learn a
consistent ordering with at most $O(n^2)$ mistakes.  We mention briefly
that this can be used directly to solve the problem of predicting $f=
I$ (the allocation is the label) for additive buyers. 

\begin{corollary}\label{cor:add}
  There is an algorithm with mistake bound $O(nm + n^2)$ for learning
  the allocation rule of an ordered arrival mechanism when buyers are
  additive.
\end{corollary}
\begin{proof}
  Consider the following algorithm. Initiate a permutation
  datastructure $P$, and also a $n\times m$-dimensional vector
  $B[i][e] = 0$. For a given subset $S^t$, in order according to $P$'s
  current permutation, the algorithm predicts that agent $i$ is
  allocated every item which is both still available and for which
  $B[i][e] = 1$. There are two types of errors which are
  made. Consider $i$, the first buyer (according to $P$) for which we
  made an error.  If $i$ wins some item $e$ which our algorithm did
  not predict, then set $B[i][e] = 1$ (it must have been previously
  $B[i][e] = 0$, or we would have allocated $e$ to $i$). If $i$ did
  not win some element $e$ that we predicted her to win, it must be
  the case that $B[i][e] = 1$, $\loc{i}{P} < \loc{i'}{P}$, and
  $\loc{i'}{\s} < \loc{i}{\s}$ for some $i'$ (namely, that $i'$ that
  won $e$). Thus, by Lemma~\ref{lem:perm-use}, it is valid to demote
  $i$. There are at most $nm$ errors of the first type and at most
  $n^2$ of the second type.
\end{proof}

We also briefly mention that this permutation datastructure is enough
to learn in the simple case that $f = I$ (the allocation is the label) for single-minded buyers.

\begin{corollary}\label{cor:single-id}
  There is an efficient algorithm for predicting $f=I$ (the
  allocation) for subsets of single-minded buyers with a mistake
  bound of $O(n^2)$.
\end{corollary}
\begin{proof}
  Use an instantiation of the permutation datastructure $P$ as
  above. For each buyer $i$, let $\hat D_i = \emptyset$
  initially. Whenever we see a buyer $i$ win a nonempty set, set
  $\hat D_i = X_i = D_i$.  When a subset $S^t$ arrives, in order
  according to $P$, allocate $j\in S^t$ his set $\hat D_j$ if it is
  still available (otherwise, $\hat X_j = \emptyset$). When this
  allocation rule makes a mistake, consider $i^t$, the first buyer
  (according to the ordering $P$) for which our algorithm mis-allocated
  items. There are two possible mistakes: $\hat X_{i^t} = \emptyset$
  but $X_{i^t} = D_{i^t}$, or $X_{i^t} = \emptyset$ but $\hat X_{i^t}
  = D_{i^t}$.  The first case can occur for two reasons: we have never
  seen $i^t$ win, in which case we have $\hat D_{i^t} = \emptyset$, so
  we then set $\hat D_{i^t} = D_{i^t}$ (there are at most $n$ of these
  errors), or because there is some $i'$ such that $i^t \s i'$ but $i'
  \sh^t i^t$ (who conflicts with $i^t$). But this is not possible, or
  $i'$ would be an earlier mistake. The second kind of mistake can
  only occur because there is some $i'$ such that $i' \s i^t $ but
  $i^t \sh^t i'$ ($i^t$ is too early in the permutation). Thus, by
  Lemma~\ref{lem:perm-use}, demoting $i^t$ in these cases is valid
  and leads to a mistake bound of $O(n^2)$. Thus, in total, there are
  $O(n^2)$ mistakes made by this algorithm.
\end{proof}

It remains to show how we can use the permutation datastructure to
solve our original problem, that of learning $\s$ alongside the
conflict graph for single-minded buyers from the examples of winning
sets. The intuition behind our main algorithm is as follows. We will
initialize the permutation datastructure and begin by assuming the
conflict graph is the complete graph. For a given estimate $\sh$ and
conflict graph $\widehat{G}$, we predict $\widehat{W}^t$ for $S^t$ as
follows. The mechanism serves members of $S^t$ in order according to
$\sh$ (e.g, $\loc{i}{P^t} < \loc{j}{P^t}$ will imply $i$ gets served
before $j$), subject to the constraint that if $j$ is in conflict with
some earlier buyer who has won, $j$ doesn't win. Then, there will be
two types of mistakes:
when $W^t$ includes some pair $(i,j)$ connected by an edge in $\widehat{G}$,
and when it does not.  In the first case, we can safely remove $(i,j)$
from $\widehat{G}$, and in the second case we will argue that
we can safely demote some buyer. We will maintain the
invariants alluded to previously: namely, that $E \subseteq
\widehat{E}$ (for edges in the conflict/current estimate graph), and
that we have never demoted a buyer who didn't need to be
demoted. Algorithm~\ref{alg:single} formalizes this intuition.

\begin{algorithm}
  \SetAlgoRefName{SingleMinded}
  P=\init\;
  Let $\widehat{G} = ([n], \widehat{E})$ where $(i,j)\in \widehat{E}$ for all $i \neq j$\; 
  \For {$t = 1$ to $T$} {
    Receive $S^t$\;
   Let $\widehat W^t = \emptyset$\;
   \For{$b= 1$ to $ |S^t|$} {
     Let $i = \ith{b}{S^t}{P}$\;
     add $i$ to $\widehat W^t$ if $\nexists j \in \widehat W^t$ such that $(i,j)\in \widehat E$\;
   }
   Predict $\widehat W^t$\;
   Learn $W^t$\;
   \If{$W^t \neq \widehat W^t$} {
     \eIf{$\exists i,j \in W^t$ such that $(i,j)\in \widehat E$} {$\widehat E = \widehat E \setminus \{(i,j)\}$ }
     {Let $i^t = \ith{1}{\widehat W^t \setminus W^t}{P}$\;
       $P = \demote{i^t}{P}$\;
     }
   }
  }
\caption{MB algorithm for predicting winners;
  single-minded buyers wrt order $\s$}\label{alg:single}
\end{algorithm}

%


\begin{thm}\label{thm:single}
  ~\ref{alg:single} is a $2n^2$-mistake bound algorithm for predicting
  $W^t$, the winning set for single-minded buyers.
\end{thm}

\begin{proof}
  Throughout the life of the algorithm, two invariants are
  maintained. First, the true set of edges in the conflict graph $E$
  will always be contained in $\widehat E$ the set of conflicts the
  algorithm tracks. Second, $\widehat \s $ will only be told to push down
  a buyer $i$ when there is some $j$ such that $\loc{i}{\s} >
  \loc{j}{\s}$ but $\loc{i}{P^t} < \loc{j}{P^t}$.

  We proceed to show the first invariant holds. It begins with the
  complete conflict graph $\widehat G = ([n], \widehat E)$, so the
  invariant holds at the beginning.  Whenever the algorithm deletes an
  edge $(i,j)$ from $\widehat G$, an example has been observed where
  two buyers are clearly not in conflict (e.g., $i$ and $j$ are both
  allocated in some example).

  Now, we prove the second invariant. This is clearly true when we
  push some $i$ down the first time; $i$ doesn't always win when he
  shows up, implying he isn't at the first level according to
  $\s $. Now, suppose so far this has been the case: no buyer so
  far has been pushed unless there was proof according to $\s $
  that he is below someone below him according to
  $\widehat\s ^{t-1}$. Then, when $i^t$ is asked to be pushed down, it
  is because $i^t$ wasn't allocated to, even though $\widehat\s ^{t-1}$
  said he should have been. This isn't because of conflicts, by
  invariant $1$, so $i^t$ didn't conflict with those above him
  according to $\widehat\s ^{t-1}$.  Moreover, since $i^t$ was \emph{the
    first person according to $\widehat\s ^{t-1}$ where we made a
    mistake}, $i^t$ must conflict with someone above him according to
  $\s $, implying there is someone below him in $\widehat\s ^{t-1}$
  who he is below in $\s $. Thus, the invariant is maintained after
  $i^t$ is demoted.

  Thus, by Lemma~\ref{lem:perm-use}, using the permutation
  datastructure is appropriate: it is never told to push down some
  buyer who doesn't need to be lower according to $\s $, so the
  algorithm is correct.  Finally, there can be at most $n^2$ edges
  deleted from the conflict graph, and at most $n^2$ times where some
  buyer is pushed down in the ordering. Thus, the mistake bound on
  this algorithm is $2n^2$.
\end{proof}

Thus, this is quite an effective way to interpret the observations
of ``satisfied'' or ``not satisfied'': if the observation was
actually the allocation (and the goal to predict the allocation),
the problem trivially reduces to the $n^2$ bound from the
single-item case (it reduces to learning the priority of each buyer
and seeing each buyer win one time). If, on the other hand, the
observations are simply the subset of those buyers who are
satisfied, and we aim to construct the smallest consistent model of
the items corresponding to the demanded sets, the problem becomes
NP-complete.
\begin{observation}
  Given a conflict graph $G$, finding the smallest $m$ for which 
 single-minded buyers over $m$ items suffices to describe the
 conflicts in $G$ is equivalent to clique edge-cover, and is thus
  NP-complete. On the other hand, there will always exist a consistent set of at most
  $n^2$ items: in particular, one item for each edge in $G$ with each
 player wanting all of its incident edges.
\end{observation}

%% file: unit.tex
\section{Unit-demand buyers}\label{sec:unit-demand}

Suppose now our $n$ buyers are unit-demand, and we wish to predict
the allocation rather than just the winning set. When prices are
fixed, this problem corresponds to each buyer having an ordering over
items $>_i$, as well as the ordering $\s$ over buyers. At each time
$t$, a subset $S^t$ arrives, and the players in $S^t$, in order
according to $\s $, will each choose their favorite item
remaining. For example, the first buyer in $S^t$ will choose his
favorite item, the second buyer will choose his favorite that the
first buyer didn't take, and so on. By a reduction to the
single-minded case, we have the following.

\begin{thm}\label{thm:unitreduction}
  Algorithm~\ref{alg:unit} is an $O(n^2m^2)$-mistake bound learning
  algorithm for predicting the allocation for subsets of buyers, when
  the allocation occurs according to some fixed permutation on
  unit-demand buyers facing fixed prices, when the observation is the
  true allocation in that setting.
\end{thm}

\begin{algorithm}[t]
\SetAlgoRefName{UnitDemand}
  Let $U$ be an instantiation of Algorithm~\ref{alg:single} for $nm$ buyers $i_{11}, i_{12}, \ldots i_{nm}$\;
  \For {$t = 1$ to $T$} {
    Receive $S^t$\;
    Let $\bar S^t = \{ i_{jk} | j\in S^t, \forall k\in [m]\}$;\ \tcp{Convert players to ghost buyers}
    Let $\bar W^t = U(\bar S^t)$\;   
    \For{$j\in S^t$} {
      Let $\hat W^t(j) = k$, where $i_{jk}\in \bar W^t$;\  \tcp{Convert ghost  winners to allocation}
    }
    Predict $\hat W^t$\;
    Learn $W^t$\;
    \If{$W^t \neq \hat W^t$} {
      Let $\hat{\overline{ W^t}} = \{i_{jk}|W^t(j) = k\}$;\ \tcp{Convert allocation to ghost winners}
      Give $U$ the mistake $\hat{\overline{ W^t}}$\;
    }
}\caption{Predicts allocation of order-based allocation rule
  for unit-demand players\label{alg:unit}}
\end{algorithm}

\begin{proof}
  Consider Algorithm~\ref{alg:single}. For each player $i$, make $m$
  ``ghost buyers'' $i_{1}\ldots i_m$, which will correspond to
  embedding $i$'s preferences into $\s $.  Notice that the true
  allocation mechanism $\M$ can be viewed as an ordering of the $mn$
  ghost buyers (an ordering over buyers, and within each buyer an
  ordering over items) where two ghost buyers are in conflict if
  either they correspond to the same buyer or correspond to the same
  item.  Since only one of the $m$ copies of a given buyer will be
  allocated to according to $\M$, we will never delete conflict edges
  between these copies, and will thus never predict two ghosts
  corresponding to the same true buyer will win simultaneously.
  Similarly, for any item $j$, since the item will never
  simultaneously be given to two buyers, we will never delete conflict
  edges corresponding to that item.  Finally, for a player in position
  $j$ according to $\s $, no more than $j$ of his ghost buyers will
  ever be seen winning. Thus, at most $j$ of his ghosts will need to
  be rearranged by $\hat \s $. In total, then, there are only
  $\min(n^2, nm)$ ghost buyers that are relevant. Then, there are at
  most $O(\min(n^4, n^2m^2))$ mistakes, by Theorem~\ref{thm:single}.
\end{proof}

We also briefly mention that we have a slightly tighter bound (which
is computationally efficient).

\begin{algorithm}[t]
  \SetAlgoRefName{UnitDemandPrime}
  Let $P$ be an instantiation of \ref{alg:all} for $n$ buyers $1,  \ldots n$\;
  \For {$i = 1$ to $m$} {
    Let $\hat \succ_i$ be an instantiation of the approximate halving algorithm (Section \ref{sec:discussion}) for learning $i$'s preference over $[m]$\;
  }
  \For {$t = 1$ to $T$} {
    Let $I = [m]$\;
    Receive $S^t$\;
    \For{$b = 1$ to $|S^t|$} {
      Let $i = \ith{b}{S^t}{P}$\;
      Let $\hat{j}_i = \ith{1}{I}{\hat \succ_i}$; \tcp{the predicted choice of item by $i$}
      Let $\hat{X^t_i} = \{\hat{j}_i\}$\;
      Let $I = I \setminus \{\hat{j}_i\}$\; 
    }
    Predict $\hat X^t$\;
    Learn $X^t$\;
    \If{$X^t \neq \hat X^t$} {
      Let $i^t = \firstmistake{X^t}{\widehat X^t}{P}$;   \tcp{index of first mistake according to $P$}
      Let $\{j_{i^t}\} = X^t_{i^t} $\;
      Give the constraint $j_{i^t} \succ_{i^t} \hat j_{i^t}$ to $\hat \succ_{i^t}$\;
      \If{$\hat \succ_{i^t}$ is infeasible} {
        \demote{i^t}{P}\;
        Reset $\hat \succ_{i^t}$\;
      }
    }
}\caption{Predicts allocation of order-based allocation rule
  for unit-demand players\label{alg:unitprime}}
\end{algorithm}

\begin{thm}\label{thm:unit-improved}
  Algorithm~\ref{alg:unitprime} has a mistake bound of
  $\Theta(n^2m\log(m))$ for predicting the allocation for subsets of
  buyers, when the allocation occurs according to some fixed
  permutation on unit-demand buyers with fixed prices when the
  observation is the true allocation in that setting.
\end{thm} 

\begin{proof}
  We describe how to efficiently implement an approximate halving
  algorithm for learning permutations consistent with a partial order
  in Section~\ref{sec:discussion}. If $i^t$ is the first mistake,
  either the estimate of $i^t$'s preferences are incorrect (in which
  case $j_{i^t} \succ_{i^t} \hat j_{i^t}$, and we add this
  constraint), or $i^t$ needs to be demoted. Once $\hat \succ_{i^t}$
  becomes infeasible, all constraints added were valid, so demoting
  $i^t$ is valid.
\end{proof}

\subsection{Multiple copies}

Several of these results are easy to extend to the setting where there
are multiple copies of each resource, and players are unit-demand for
multiple copies of a particular item (e.g., no player wants more than
one copy of a given item). If buyers are additive (across bundles of
different items), it suffices to learn their preferences over
\emph{types} of items. This can be done as in Corollary~\ref{cor:add}
with no loss, treating any copy of a resource identically internal to
the learning algorithm. For prediction, a buyer will take his favorite
\emph{bundle} of items, and that bundle can contain an item for which
there is at least one copy remaining. This implies a mistake bound for
learning the allocations which is \emph{independent} of the number of
copies of each item.

For unit-demand buyers, it is not clear how to use the solution from
Theorem~\ref{thm:unitreduction}, which reduces to the single-minded
case (for buyer $i$ to not have item $e$ available, there would need
to be $k_e$ ``ghosts'' that bought item $e$ prior to buyer $i$, rather
than a single conflict). On the other hand,
Algorithm~\ref{alg:unitprime} can be used directly, to learn the
permutation over buyers and, for each buyer $i$, $i$'s preference
order over item types.

In the case of single-minded buyers, recall that the problem is quite
easy if we ever see an allocation; Corollary~\ref{cor:single-id}
applies directly with no loss in the mistake bound.  On the other
hand, if no allocation is seen, and instead we only see the subset of
people who received their set, one can use a \emph{conflict
  hypergraph} rather than a conflict graph. The total number of edges
in a necessary hypergraph blows up rather quickly, unfortunately: the
size of this representation (and thus the number of mistakes) will be
$\Theta(n^k)$ where $k = \max_{j\in m}k_j$. Based on our previous
observation about the complexity of finding a minimal representation
(in terms of items) consistent with the perceived conflicts, even when
there is only one copy of each item, we suspect this problem may be
inherent. We leave the question of whether one can predict the winning
sets of single-minded buyers with a mistake bound and running time
which is polynomial in $m,n$ and $k$ (or even independent of $k$, for
the mistake bound).

%% file: prices.tex
\section{Variable prices: additive and unit-demand}\label{sec:prices}

The previous sections can be thought of simulating a simple mechanism:
according to the mechanism, buyers have priorities, and in order of
that priority, buyers will pick their favorite bundle available. For
single-minded buyers, the priorities could be thought of as a sorting
of buyers by their bid, or some other pecking order. A buyer's
preferences could be thought of as an ordering according to value, or
quasilinear utility according to some fixed prices. We now consider a
twist on this original setup: what if, rather than the prices being
fixed, each round was fed a price vector $p^t$ along with the subset
$S^t$, with the assumption that buyers would now take a bundle to
maximize their quasilinear utility with respect to these varying
prices? Assume, for simplicity, that $p^t \in \{0,1,\ldots, V\}^m$.
Simply running our previous mistake bound algorithm for unit-demand,
additive, or single-minded buyers is tantalizingly simple. Since
buyer's preferences over bundles will change from one round to the
next, this approach fails miserably.

The algorithm which solves the fixed price problem for unit-demand
buyers can be thought of in a slightly different way, which will be
useful for solving the problems with variable prices. An equivalent
solution to the problem is to start with a permutation datastructure
$P$ to learn the ordering over buyers, and for each buyer $i$, a
permutation datastructure $P_i$ to learn their preference ordering
over items. Whenever a mistake is made, the algorithm blames the
subroutine which is learning some buyer's preferences (namely, the
earliest buyer for which we made a mistake). If this causes their
subroutine to become infeasible, it must be the case that the buyer
needs to be demoted in the larger ordering, so we demote the buyer and
restart their subroutine. Then, the total mistake bound for the
algorithm will be $n^2 \MB$, where $\MB$ is the mistake bound for the
subroutines (because each buyer can be demoted at most $n$ times, and
there are at most $\MB$ mistakes for a buyer at each position).  This
intuition (running a global algorithm with subroutines for each
buyer's preferences) is our starting point for constructing
mistake-bound learning algorithms for variable prices. We begin by
designing an algorithm for additive buyers, which gives some intuition
for the case of unit-demand. We assume, for simplicity, there are no
ties for a buyer's most-preferred bundle at any set of prices.  All
results can be extended to allow ties assuming buyers break ties
consistently.

\subsection{Additive}\label{sec:add-prices}

Suppose in each round $t$, our input is a subset of buyers present
$S^t$ and a price vector $p^t$. The algorithm $\mathcal{A}$ we are
trying to simulate is composed of two parts, $\s $, a priority over
buyers, and $v_i(j)$ for each buyer $i\in B$, and $ j\in I$,
corresponding to the value buyer $i$ has for item $j$. On a given
subset and price vector pair, $\mathcal{A}$ will offer all items to
buyer $i\in S^t$ who is first in $\s $, who will take all items such
that $v_i(j) > p^t(j)$. Then, the remaining items are offered to the
remaining buyers $i'\in S^t$, in order of $\s$, who will do the same.

Our algorithm will predict an allocation of item $\widehat X^t_1,
\ldots, \widehat X^t_n$. If the allocation is wrong, the correct
allocation $X_1^t, \ldots, X_n^t$ is shown to our algorithm. We wish,
for arbitrary price vectors and subsets, to minimize the total number
of days on which we make prediction errors.

The argument of correctness for Algorithm~\ref{alg:additivevariable} (below)
is somewhat more complex than in the previous sections. As before, we
need to show that the algorithm never tells $P$ to push down a buyer
when it should not. The condition for this is slightly more
complicated, however. That is, the infeasibility of the binary search
for $v_i(j)$ is proof that this buyer cannot be above all of the
buyers below him according to $P$. It is important that the first
error according to $\perm{P}$ (and \emph{only} this error) is the one
used to update the model: this avoids an earlier error in the ordering
$P_t$ (or an error in an earlier binary search $v_i(j)$) placing
incorrect constraints on lower buyers.  To this end, let
$\firstmistake{X}{X'}{P}$ denote the first $i$, according to $P$, for
which $X_i \neq X'_i$. Let $V$ be the maximum valuation any buyer
has for any item.

\begin{algorithm}[t]
  \SetAlgoRefName{AdditiveVariable}
  $P = \init$\;
  Let $\bar v_i(j) = V$; $\underline v_i(j) = 0$\;
  \For {$t = 1$ to $T$} {
   Receive $S^t, p^t$\;
   Let $\widehat X_i^t = \emptyset$; $I' = [m]$\;
   Let $\widehat v^t_i(j)= \frac{\bar v_i(j) + \underline v_i(j)}{2}$\;
   \For{$b= 1$ to $ |S^t|$} {
     Let $i = \ith{b}{S^t}{P}$;  *Consider the $b$th-ranked buyer of $S^t$ according to $P$*\\
     Let $X_{i}^t = \{j \in M | \widehat v_i(j) > p^t(j)\}$ and  $I' = I' \setminus X_{i}^t$\;
   }
   Predict $\widehat X_1^t, \ldots, \widehat X_n^t$\;
   Learn $X_1^t, \ldots, X_n^t$\;
   \If{$X^t \neq \widehat X^t$} {
    Let $i^t = \firstmistake{X^t}{\widehat X^t}{P}$;   *index of first mistake according to $P$* \\
    \eIf{$\exists j\in X^t_{i^t}\Delta \widehat X^t_{i^t}$ such that $\overline v_{i^t}(j) = \underline v_{i^t}(j)$  *$\exists$ item for which binary search is invalid* } {
      $\demote{i^t}{P}$, Let $\underline v_{i^t}(j) = 0$ and $\bar v_{i^t}(j) = V$\;
      }
      {* Update the binary searches. *\\
        \For{each $j\in  \widehat X^t_{i^t}\setminus  X^t_{i^t}$} {
          Set $\overline v_{i^t}(j) = \widehat v_{i^t}(j)$\;
        }
        \For{each $j \in  X^t_{i^t}\setminus \widehat  X^t_{i^t}$} {
          Set $\bar v_{i^t}(j) = \widehat v_{i^t}(j)$\;
        }
      }
    }
  }
\caption{MB algorithm predicting $X_1^t, \ldots X_n^t$; additive
  buyers under order $\s$ \label{alg:additivevariable}}
\end{algorithm}

\begin{thm}\label{thm:additivevariable}
  Algorithm~\ref{alg:additivevariable} is an $O(\log(V)m n^2)$
  mistake-bound learning algorithm for the problem of predicting
  subsets of quasilinear additive buyer's purchases according to $\s$
  priority with prices $p^t$.
\end{thm}

We prove the following lemma as a starting point for
Theorem~\ref{thm:additivevariable}.

\begin{lemma}\label{lem:pos}
  Let buyer $i^t\in B$ be the first mistake in round $t$. Then, there
  exists some item $j\in I$ for which one of these statements holds:
\begin{enumerate}
\item  $j\notin \widehat X_{i^t}^t$, but $j\in
  X_{i^t}^t$, and $\widehat v_{i^t}(j) < v_{i^t}(j)$
\item $j\in \widehat X^{i^t}_t$, but $j\notin
  X_{i^t}^t$, and $\widehat v_{i^t}(j) > v_{i^t}(j)$
\item $j\in \widehat X_{i^t}^t$, but $j\notin X_{i^t}^t$, and there is
  some $i'$ such that $\loc{P^t}{i^t} < \loc{P^t}{i'}$ but
  $\loc{\s}{i^t} > \loc{\s}{i'}$.
\end{enumerate}
\end{lemma}

\begin{proof}
  A mistake implies either there was some item $j\notin \widehat
  X_{i^t}^t$ but $j\in X_{i^t}^t$, or $j\in \widehat X_{i^t}^t$ but
  $j\notin X_{i^t}^t$.

  Consider the first case.  Our algorithm did not give item $j$ to buyer $i^t$:
  moreover, it didn't give item $j$ to some buyer $i'$ with higher priority
  ($\loc{P^t}{i'} < \loc{P^t}{i^t}$), since buyer $i^t$ was the \emph{first}
  mistake, implying our estimate of buyer $i^t$'s value for item $j$ was too low,
  which falls into case $1$.

  If, on the other hand, our algorithm allocated item $j$ to buyer $i^t$, but
  buyer $i^t$ was not awarded item $j$, then either buyer $i^t$'s value for item $j$ is less
  than the price (and our estimate $\widehat v_{i^t}(j)$ was too high,
  implying case $2$), or some buyer $i'$ with higher priority w.r.t $\s$ took item $j$
  before buyer $i^t$. In this case, since buyer $i^t$ was the first mistake, this
  implies $\loc{P^t}{i^t} < \loc{P^t}{i'}$ but $\loc{\s}{i^t} >
  \loc{\s}{i'}$, implying case $3$.
\end{proof}

Now, we prove Theorem~\ref{thm:additivevariable}.

\begin{proof}
  Along the lines of our previous proofs, we will show our algorithm
  maintains two invariants:
\begin{enumerate}

\item For any buyer $i$, for each item $j$, at each time $t$, the binary
  search for $v_i(j)$ has been given only accurate upper and lower
  bounds for any permutation $\s'$ such that $i$ has not been demoted
  in $\s'$ from $\widehat\s^t$ (but some other buyers may be
  demoted). That is, if buyer $i$'s precedence does not decrease, any
  $v_i(j)$s which would be consistent with $\s'$ and the observations
  are consistent with the binary searches and $\widehat\s^t$.
\item Any buyer $i^t$ demoted in $P_t$ has some buyer $i'$ such that
  $\loc{P^t}{i^t} < \loc{P^t}{i'}$ but $\loc{\s}{i'} < \loc{\s}{i^t}$.
\end{enumerate}

We begin with the first invariant. It is satisfied prior to any
constraints being added to any buyer's binary searches. Now, suppose
it is true until time $t$: all constraints are accurate w.r.t
$\widehat\s^{t}$ and any $\s'$ such that $i$ has not been demoted from
$\widehat{\s^t}$ to $\s'$ but other buyers may have been demoted. If, at
time $t$, $i$ gets another constraint added to her binary searches,
this implies either this constraint is correct or buyer $i$ needs to be
demoted, by Lemma~\ref{lem:pos}. Thus, if buyer $i$ is not demoted (as is
the case for $\s'$), this new constraint (and so all the constraints)
in her binary searches are valid.

Now, we prove the second invariant.  The invariant is true at the
beginning of the algorithm prior to any buyer being pushed
downwards. Now, consider some time $t$, and assume this invariant
holds until time $t$. The only case of interest is when one of buyer
$i^t$'s binary searches is infeasible and she is demoted.

Due to invariant 1, we know that all the constraints buyer $i^t$ has
received until time $t$ are valid at her position in $\widehat \s^t$ (or
any earlier position), since buyer $i^t$'s binary searches are reset
whenever buyer $i^t$ is pushed down. Since she is the first mistake, by
Lemma~\ref{lem:pos}, it is either the case that the current constraint
being added is true with respect to her $v_{i^t}(j)$s and her
position (or an earlier position), or she must occur later in the
ordering. Thus, all the constraints in her binary searches are correct
with respect to her current position (or any earlier one) in the
ordering. Since there are some $v_{i^t}(j)$s which are consistent with
the observations, but not the set of constraints, it must be the case
that buyer $i^t$ occurs later in the ordering.

Now, we show the mistake bound. If a mistake is made, some buyer $i^t$
either updates her binary searches or is demoted. At most $\log(V)$
binary search updates can occur for a given item and buyer before the
binary search becomes infeasible and she is demoted. Thus, there can
be at most $m \log(V)$ mistakes resulting in binary search updates for
a buyer before she is demoted. By Lemma~\ref{lem:perm-use} and
invariant 2, no buyer is pushed later in the ordering than she occurs
in $\s$; thus, there are at most $n^2$ mistakes resulting in
demotions. Thus, in total, there are at most $n^2 m \log(V)$ many
mistakes.\end{proof}

\subsection{Unit-Demand}\label{sec:unit-prices}

The case of unit-demand buyers is similar to that of additive buyers,
though the buyers will no longer have \emph{separable} preferences
over items: instead, out of a set of available items $T$ at prices
$p^t$, buyer $i$ will buy $j = \argmax_{j\in J} [v_i(j) - p^t(j)]$ to
maximize his quasilinear utility (assume that there is some consistent
tie-breaking in the event that several items are equally good). So,
rather than using binary search for each item separately, for each
buyer, we will run a mistake bound ellipsoid algorithm; whenever a
constraint is added, it will be of the form $v_i(j) - p^t(j) > v_i(j')
- p^t(j')$, where the $v_i(j)$s are variables and the $p^t(j)$s are
constants coming from the online price vectors.

\begin{algorithm}[t]
\SetAlgoRefName{UnitVariable}
  $P = \init$\;
  Let $\widehat R_i$ be an instance of an ellipsoid algorithm w. unknowns $v_i(j)$ for all $i \in B, j\in I$\;
  \For {$t = 1$ to $T$} {
   Receive $S^t, p^t$\;
   Let $\widehat X_i^t = \emptyset$ an $I' = [m]$\;
   \For {$b = 1$ to $|S^t|$} {
     Let $i = \ith{b}{S^t}{P}$\;
     Let $\widehat v_i\in \mathbb{R}^m$ be the center estimated by $\widehat R_i$\;
     Let $\widehat j_i = \argmax_{j\in I'} \widehat v_i(j) - p^t(j)$;  // prediction for $i$ w.r.t. prices, est. values, remaining items \\
     Let $X_{i}^t = \{\widehat j_i\}$ if $\widehat v_i(\widehat j_i) - p^t(\widehat j_i) > 0$ or $\emptyset$ otherwise\;
     Let $I' = I' \setminus X_i^t$\;
   }
   Predict $\widehat X_1^t, \ldots, \widehat X_n^t$\;
   Learn $X_1^t, \ldots, X_n^t$\;
   \If{$X \neq \widehat X$} {
     Let $i^t = \firstmistake{X^t}{\widehat X^t}{P}$\;
     \eIf{$X_{i^t}^t = \emptyset$} {
       Give the constraint $v_{i^t}(\widehat j_{i^t}) < p^t(j_{i^t})$ to $\widehat R_{i^t}$\; }
     {  Let $\{j_{i^t}\} = X_{i^t}^t$; // The item $i^t$ actually won\\
        \eIf{$\widehat X_{i^t}^t = \emptyset$} {
          Give the constraint $v_{i^t}( j_{i^t}) > p^t(j_{i^t})$ to $\widehat R_{i^t}$\; }
        { Let $\widehat X_{i^t}^t = \{\widehat j_{i^t}\}$\;
          Give the constraint $v_{i^t}( j_{i^t}) - v_{i^t}(\widehat j_{i^t}) > p^t(j_{i^t}) - p^t(\widehat j_{i^t})$ to $\widehat R_{i^t}$\;
        }
        \If{$\widehat R_{i^t}$ is infeasible} {
          $\demote{i^t}{P}$ and restart $\widehat R_{i^t}$\;
        }
      }
   }
}
\caption{MB algorithm predicting $X_1^t, \ldots X_n^t$; unit-demand
  buyers, order $\s$}\label{alg:unitvariable}
\end{algorithm}

\begin{thm}\label{thm:unitvariable}
  Algorithm~\ref{alg:unitvariable} is an $O(n^2\MB )$-mistake bound learner
  for unit-demand buyers with respect to some order $\s$, where $\MB$
  is the online mistake bound guarantee of the online classification
  algorithm.
\end{thm}

The main theorem of this section follows from a similar analysis to
that of additive buyers in the previous section, with a twist stemming
from the fact that we use the Ellipsoid algorithm as the mistake-bound 
subroutine (with each mistake serving as its separation
oracle), rather than binary search for each item separately. 
This is similar to the use of the Ellipsoid algorithm
by \citet{maass1990fast} for learning
a linear separator.  We start by
stating a lemma about the mistake bound for this subroutine.
%
%
%
\begin{lemma}\label{lem:unit-half}
  Using the ellipsoid algorithm to learn the collection $\{v_i(j)\}_j$ has
  a mistake bound of $O(m^2(K+\log m))$ so long as each mistake returns a 
  constraint such that its current hypothesis $\hat v_i$ is no longer
  feasible, where $K$ is the maximum precision of the $v_i$s.
\end{lemma}
\begin{proof}
We will have $m$ variables corresponding to the valuations of buyer $i$ to each of the $m$ items.
The Ellipsoid algorithm maintains an ellipsoid that contains the feasible region (the possible $m$-tuples of valuations
consistent with observations so far) and proposes as its current hypothesis the center of that ellipsoid.
%
 We use this center as a proposed valuation for buyer $i$ until
 we make an error involving her. Once we make an error we identify a violated linear constraint, and we return it to the Ellipsoid algorithm, which then updates its ellipsoid and hypothesis.

 In each iteration (mistake of the algorithm) the volume shrinks multiplicatively by a fraction of  $1 - \frac{1}{m}$.
 The initial volume is at most $2^{O(m(K+\log m))}$. The final volume, assuming that there is a consistent valuation,
 is at least $2^{-O(m(K+\log m))}$. This implies that after at most $O(m^2(K+\log m))$ errors
 we reach a volume which is too small, and therefore we can declare there is no feasible valuation.
%
\end{proof}

We now state the analogue to Lemma~\ref{lem:pos} for the unit-demand
case.

\begin{lemma}\label{lem:pos-unit}
  Suppose~\ref{alg:unit} makes a mistake at time $t$. Let $i^t$ be the
  first mistake (according to $\widehat\s^t$). Let $\widehat j_{i^t}$
  be the item we predicted $i^t$ to win (if any) and $j_{i^t}$ the
  item $i^t$ won (if any). Then one of these holds:
\begin{enumerate}
\item $v_{i^t}(\widehat j_{i^t}) - p^t(\widehat j_{i^t}) < 0 < \widehat v_{i^t}(\widehat j_{i^t}) - p^t(\widehat j_{i^t})$, or $ v_{i^t}(\widehat j_{i^t}) < \widehat v_{i^t}(\widehat j_{i^t})$
\item $v_{i^t}(\widehat j_{i^t}) - p^t(\widehat j_{i^t}) < v_{i^t}(j_{i^t}) - p^t(j_{i^t})$
\item $v_{i^t}(j_{i^t}) > \widehat v_{i^t}(j_{i^t})$
\item $\widehat j_{i^t}$ was not available ($\exists i'$ s.t.
  $\loc{P^t}{i^t} < \loc{P^t}{i'}$ but  $\loc{\s}{i^t} >
  \loc{\s}{i'}$.
\end{enumerate}
\end{lemma}

\begin{proof}
  Consider a mistake on buyer $i^t$. Either it is the case that (a)
  $i^t$ bought nothing and we predicted she bought something, (b) she
  bought something and we predicted nothing, or (c) we predicted the
  wrong item.

  (a) occurs only when $\widehat j$ was no longer available ($i^t$ needs
  to be demoted, case 4) or $\widehat j$ was too expensive, $v_{i^t}(\widehat
  j) - p^t(\widehat j) < 0$ (case 1).  (b) can only occur because our
  estimate of her value of an item was to small (case 3), since she is
  the first mistake it cannot be because we predicted that someone
  earlier took $j$. (c) occurs when either $\widehat j$ was not available
  (and $i$ needs a demotion, case 4) or our estimate of utility was
  wrong (case 2).\end{proof}

Now, we prove Theorem~\ref{thm:unitvariable}.

\begin{proof}
  We claim the same two invariants are true of
  Algorithm~\ref{alg:unitvariable} as were true of
  Algorithm~\ref{alg:additivevariable}, since Lemma~\ref{lem:pos-unit}
  provides the analogous guarantees (namely, that when we make a
  mistake, we either get to add a constraint to some ellipsoid
  algorithm, or we get to demote some buyer). Thus, the algorithm is
  correct.  Each instantiation of the ellipsoid algorithm makes at
  most $\MB$ mistakes before it is demoted and restarted, and there
  are at most $n^2$ demotions total.
 Thus, a mistake bound of $\MB n^2$ in
  total holds.\end{proof}

%% file: discussion.tex
\section{Discussion}\label{sec:discussion}
In this paper we present algorithms that from observations of opaque
transactions (observing just who wins and who doesn't in the case of
single-minded buyers, or observing the allocations produced in the
case of additive or unit-demand buyers) can reconstruct  both the
preferences of the buyers and the mechanism used by the seller
sufficiently well to predict the outcomes of new transactions.  We
focus on priority-based {\em ordered arrival mechanisms} on the side
of the seller, and commonly-studied classes of valuation functions for
the buyers.  It would be interesting to consider this problem in the
context of other mechanisms and other observation models as well.
Note that for mechanisms such as VCG (producing a
social-welfare-maximizing allocation) certain complications arise: for
instance even in the case that all buyer valuations are known, finding
the allocation produced can be NP-complete if buyers are
single-minded.  So one would want to focus on settings where at least
when everything is known the prediction problem is easy.

A concrete open question is whether one can improve the mistake bounds
given in the previous sections with computationally efficient
algorithms. While the single-minded mistake bound has a matching
information-theoretic lower bound\footnote{Suppose there are
  $\frac{n(n-1)}{2}$ items (one for each pair of buyers). In each
  round $t$, the adversary presents the algorithm with $S^t$ which
  contains a pair of buyers that has never been presented before. The
  algorithm needs to predict whether one or both of the
  buyers will be satisfied (guessing whether both buyers are both
  interested in their ``shared'' item or not). Regardless of the
  algorithm's choice, the adversary will say that was a mistake: this
  yields a consistent set of conflicts and will force the algorithm to
  make $\Omega(n^2)$ mistakes.}, there may be room for improving the
unit-demand results.  Information theoretically, there are matching
upper and lower bounds for several of these problems. As a warm-up, we
first present the single item case (where we {\em do} have an
efficient algorithm for the matching the upper bound). We state the
formal theorem below.

\begin{thm}\label{thm:single}
  The problem of learning the allocation made by an ordered arrival
  mechanism with a single item has a mistake bound $M =
  \Theta(n\log(n))$, and there is an algorithm with this mistake bound
  that runs in polynomial time.
\end{thm}
For the lower bound, an adversary can present subsets of size 2 and
essentially just simulate merge-sort.  To start, for $i=1,\ldots,
n/2$, the adversary presents subset $\{2i-1,2i\}$, and tells the
algorithm it has made a mistake (regardless of its prediction),
causing $n/2$ mistakes.  In general, given $n/L$ sorted lists of size
$L$, the adversary pairs the lists together and then for each pair
runs through the merging process (presenting the subset consisting of
the top element in each list, telling the algorithm it has made a
mistake whatever its prediction is, and popping off the true largest
element).  This maintains consistency with an overall ordering and
creates at least $L$ mistakes per pair, or again $n/2$ mistakes total
for the round.  There are $\log(n)$ rounds, leading to an overall
lower bound of $\Omega(n \log n)$.

We can construct a computationally efficiently algorithm which matches
this information-theoretic lower bound using two ideas. First, each
mistake gives us a new pair of agents $(i,j)$ for which we learn $i \s
j$ but $\loc{i}{P^t} < \loc{j}{P^t}$ (the true winner $i$ has higher
priority than every other $i'\in S^t$, and in particular, the
estimated winner $j$). Second, as mentioned previously,
~\citet{karzanov1991conductance} given an efficient sampling algorithm
which samples uniformly a consistent linear extension of a partial
order.

Then, consider the following prediction algorithm. Consider a new
subset $S^t$. Take a single sample $\s_?$ using the algorithm
of~\citet{karzanov1991conductance}, and predict the winner is $j^t =
\ith{1}{S^t}{\s_?}$. If a mistake is made, and $i^t$ is the winner,
add the set of constraints $i^t \s j$ for all $j\in S^t$ to the
partial order. We claim that each constraint added to the partial
order over the life of the algorithm is correct (they are added
because a mistake is proof of the constraint). Second, when a mistake
is made, the number of consistent linear extensions shrinks
(multiplicatively) by at least $\frac{1}{4}$ in expectation. This fact
follows from the fact that if there is some $k^t$ whose probability of
winning at time $t$ is at least $\frac{1}{2}$ (where this probability
is taken over the set of consistent linear extensions), there is
probability at least $\frac{1}{2}$ of our algorithm predicting
$k^t$. If $k^t$ is incorrect, then all permutations where $k^t$ is
first amongst $S^t$ are inconsistent after adding the new constraints,
cutting the number of consistent linear extensions in half. Another
winner is predicted with probability at most $\frac{1}{2}$, and the
set of linear extensions only shrinks. Thus, by an analysis similar to
the halving algorithm, after $\Theta(n\log(n))$ mistakes, there is
only one consistent linear extension, and it is $\s$.

The case of unit-demand buyers also has matching
information-theoretic lower and upper bounds, though we do not know of
a polynomial-time algorithm which achieves this mistake bound.

\begin{thm}\label{thm:unit-info}
  For the fixed-price problem of learning an ordered allocation
  mechanism over unit-demand buyers, the mistake bound is
  $\Theta(mn\log(m))$ (assuming $m = \Omega(\log(n))$).
\end{thm}

The lower bound for this problem is similar to the previous
argument. The generalization uses $n$ buyers, the first $m$ of which
are ``dummy'' buyers and have favorite items $a_1, \ldots, a_m$. We
can use these first buyers to control which items are available for
the true buyers. Then, each example $S^t$ will contain $m-2$ ``dummy''
buyers (who take all but just $2$ items $a_f,a_g$) and one true buyer
$i$. Then, the algorithm needs to decide which of $a_f$ or $a_g$ the
true buyer will select. This will be repeated for each pair of items
and each non-dummy buyer.  Thus, the algorithm is solving $n-m$
separate instances of sorting $m$ items (for each buyer), and so the
problem has a lower bound of $\Omega(mn\log(m))$ mistakes.

Without computational constraints, we can construct an algorithm with
a matching mistake bound.  The algorithm will maintain a list of
consistent permutations over buyers (and, for each of those
permutations over buyers, the consistent permutations for each buyer
over items). Given a new subset $S^t$, the algorithm predicts the most
likely allocation (where each consistent predictor votes once). Since
there are $n! * (m!)^n$ many initial hypotheses (an ordering over
buyers and, for each buyer, an ordering over items), the halving
algorithm will make $O(n\log(n) + nm\log(m))$ mistakes.

It is not clear how to make this algorithm computationally efficient
without increasing the mistake bound: unlike in the single-item case,
there isn't a clear culprit to our mistake. In the single-item case,
we can add another constraint to our partial order, generating a
refined partial order. In the unit-demand case, a mistake could be
made either because the understanding of some individual's preferences
are wrong, or because they were given an incorrect priority. In our
implementation, we blame the understanding of a buyer's preferences
for as long as possible. Once a buyer $i$'s preference learner is
infeasible, the algorithm has proof that some buyer $j\in S^t$ has
higher rank than $i$ (rather than one particular $j\in S^t$). We do
not know how to maintain this information as a partial order, or in
some other compact way that allows us to sample efficiently from the
linear extensions of our observations. We leave it as an open question
whether or not there is an algorithm $\mathcal{A}$ which predicts an
ordered arrival mechanism with fixed prices for unit-demand buyers
whose mistake bound is $O(nm\log(m))$ with $poly(n,m)$ computational
complexity.